\documentclass[letterpaper, 10 pt, conference]{ieeeconf}

\newtheorem{theorem}{Theorem}

\newtheorem{example}[theorem]{Example}

\newtheorem{lemma}[theorem]{Lemma}

\newtheorem{problem}[theorem]{Problem}

\newtheorem{remark}[theorem]{Remark}

\IEEEoverridecommandlockouts
\usepackage{graphicx} % for pdf, bitmapped graphics files
\usepackage{graphics}
\usepackage{epsfig}
\graphicspath{{figures/}}

\usepackage{amsmath}
\usepackage{amssymb}

\usepackage{algorithm}
\usepackage{algorithmic}
\usepackage{multirow}
\usepackage{rotating}
\usepackage{subfigure}
\usepackage{color}
\usepackage[space]{cite}
\usepackage{mathptmx}
\usepackage{subfigure}
\usepackage{times}

\begin{document}

\title{{\LARGE \textbf{Spectral Control of Mobile Robot Networks}}}
\author{Michael~M.~Zavlanos,~Victor~M.~Preciado and Ali Jadbabaie\thanks{%
Michael M. Zavlanos is with the Dept. of Mechanical Engineering, Stevens
Institute of Technology, Hoboken, NJ 07030, USA, \texttt{\footnotesize %
michael.zavlanos@stevens.edu}. Victor. M. Preciado and Ali Jadbabaie are
with the Dept. of Electrical and Systems Engineering, University of
Pennsylvania, Philadelphia, PA 19104, USA, \texttt{\footnotesize %
\{preciado,jadbabai\}@seas.upenn.edu}.}}
\maketitle

\begin{abstract}
The eigenvalue spectrum of the adjacency matrix of a network is closely
related to the behavior of many dynamical processes run over the network. In
the field of robotics, this spectrum has important implications in many
problems that require some form of distributed coordination within a team of
robots. In this paper, we propose a continuous-time control scheme that
modifies the structure of a position-dependent network of mobile robots so that it achieves a
desired set of adjacency eigenvalues. For this, we employ a novel abstraction of the
eigenvalue spectrum by means of the adjacency matrix spectral moments. Since the eigenvalue spectrum is uniquely determined by its spectral moments, this abstraction provides a way to indirectly control the eigenvalues of the network. Our construction is based on artificial potentials that capture the distance of the network's spectral moments to their desired values. Minimization of these potentials is via a gradient descent closed-loop system that, under certain convexity assumptions, ensures convergence of the network topology to one with the desired set of moments and, therefore, eigenvalues. We illustrate our approach in nontrivial
computer simulations.
\end{abstract}

%%%%%%%%%%%%%%%%%%%%%%%%%%%%%%%%%%%%%%%%%%%%%%%%%%%%%%%%%%%%%%%%%%%%%%%%%

%%%%%%%%%%%%%%%%%%%%%%%%%%%%%%%%%%%%%%%%%%%%%%%%%%%%%%%%%%%%%%%%%%%%%%%%%

\section{Introduction}

\label{sec_introduction}

A wide variety of coordinated tasks performed by teams of mobile robots
critically rely on the topology of the underlying communication network and
its spectral properties. Examples include, cooperative manipulation \cite%
{Lynch1996, Mataric1995, Wang2002}, surveillance and coverage \cite%
{Cortes2004, Parker2002, Poduri2004}, distributed averaging \cite%
{Jadbabaie2003, Saber2007, Ren2005}, formation control \cite{Desai2001,
Lawton2005, Mesbahi2001}, flocking \cite{Saber2006, Tanner2007}, and
multi-robot placement \cite{Smith2007, Zavlanos2008, Kloder2006}, that all
require some form of network connectivity, structure and, oftentimes,
spectral properties. In this paper we address the problem of controlling a
network of mobile robots to a topology with a desired eigenvalue spectrum.
This effort is a first step towards the design of controllers that allow
robots to perform their assigned tasks, while optimizing coordination within
the team.

The eigenvalue spectra of a network provide valuable information regarding
the behavior of many dynamical processes running within the network \cite%
{P08}. For example, the eigenvalue spectra of the Laplacian and adjacency
matrices of a graph affects the mixing speed of Markov chains \cite{Ald82},
the stability of synchronization of a network of nonlinear oscillators \cite%
{PC98,PV05}, the spreading of a virus in a network \cite{DGM08,PJ09}, as
well as the dynamical behavior of many decentralized network algorithms \cite%
{L97}. Similarly, the second smallest eigenvalue of the Laplacian matrix
(also called spectral gap) is broadly considered a critical parameter that
influences the stability and robustness properties of dynamical systems that
are implemented over information networks \cite{OM04,FM04}. Optimization of
the spectral gap has been studied both in a centralized 
\cite{GMS90,GB06,KM06} and decentralized context \cite{DJ06}.

In this paper, we propose a novel framework to control the structure of a
network of mobile robots to achieve a desired eigenvalue spectrum. In
particular, we focus on the spectrum of the weighted adjacency matrix of the
network, with weights that are (decreasing) functions of the inter-robot
distances. This construction is relevant, for example, in modeling the
signal strength in wireless communication networks. Although our framework
performs well with different distance metrics, in this paper, we focus on
the $\ell _{1}$ norm (Manhattan distance) primarily for analytical reasons,
since its composition with convex functions preserves their convexity.
Additionally, this metric has potential applications in indoor navigation
where the presence of obstacles forces the signals to propagate in grid-like
environments.

We employ a novel abstraction of the eigenvalue spectrum in terms of the associated spectral moments, and define artificial potentials that capture the distance
between the network's spectral moments and their desired values. These potentials are minimized via a gradient descent
algorithm, for which we show convergence to the globally optimal moments. Since the eigenvalue spectrum is uniquely determined by the associated spectral moments, our approach provides a way to indirectly control a network's eigenvalues. This work is related to \cite%
{Preciado2010}, which addresses a similar problem for static robots in
discrete environments. This formulation, however, is more
appropriate for robotics applications, where communication depends
continuously on the robot motion.

The rest of this paper is organized as follows. In Section~\ref{sec_
notation}, we introduce some graph-theoretical notation and useful results.
In Section \ref{sec_problem}, we formulate the control problem under
consideration. We introduce an artificial potential and derive the
associated motion controllers in Section \ref{sec_control} and discuss
convergence of our approach in Section \ref{sec_convergence}. Finally, in
Section~\ref{sec_simulations}, we illustrate our approach with several
computer simulations.

%%%%%%%%%%%%%%%%%%%%%%%%%%%%%%%%%%%%%%%%%%%%%%%%%%%%%%%%%%%%%%%%%%%%%%%%%

\section{Preliminaries \& Problem Definition}

\label{sec_problem_preliminaries}

%%%%%%%%%%%%%%%%%%%%%%%%%%%%%%%%%%%%%%%%%%%%%%%%%%%%%%%%%%%%%%%%%%%%%%%%%

\subsection{Notation and Preliminaries}

\label{sec_ notation}

In this section we introduce some nomenclature and results needed in our
exposition. Let $\mathcal{G}=\left( \mathcal{V},\mathcal{E},\mathcal{W}%
\right) $ denote a weighted undirected graph, with $\mathcal{V}=\left[ n%
\right] $ being a set of $n$ nodes, $\mathcal{E}\subseteq \mathcal{V}\times 
\mathcal{V}$ a set of $e$ undirected edges, and $\mathcal{W}\in \mathbb{R}%
_{+}^{e}$ a set of weights associated to the edges. If $\left\{ i,j\right\}
\in \mathcal{E}$ we call nodes $i$ and $j$ \emph{adjacent} (or neighbors),
which we denote by $i\sim j$. In this paper, we consider graphs without
self-loops, i.e., $\left\{ i,i\right\} \not\in \mathcal{E}$ for all $i\in 
\mathcal{V}$. We denote by $a_{ij}=a_{ji}\in \mathbb{R}_{+}$ the weight
associated with edge $\left\{ i,j\right\} \in \mathcal{E}$, and assume that $%
a_{ij}=0$ for $\left\{ i,j\right\} \not\in \mathcal{E}$. We define a \emph{%
walk} $\mathbf{w}$ of length $k$ from node $v_{0}$ to $v_{k}$ to be an
ordered sequence of nodes $\left( v_{0},v_{1},...,v_{k}\right) $ such that $%
v_{i}\sim v_{i+1}$ for $i=0,1,...,k-1$. We define the weight $a_{\mathbf{w}}$
of a walk $\mathbf{w}=\left( v_{0},v_{1},...,v_{k}\right) $ as $a_{\mathbf{w}%
}=\prod_{i=0}^{k-1}a_{v_{i},v_{i+1}}$.

A weighted graph can be algebraically represented via its weighted \emph{%
adjacency} matrix, defined as the $n\times n$ symmetric matrix $A_{\mathcal{G%
}}=[a_{ij}]$, where $a_{ij}$ is the weight of edge $\left\{ i,j\right\} $.
In this paper we are particularly interested in the spectral properties of $%
A_{\mathcal{G}}$. Since $A_{\mathcal{G}}$ is a real symmetric matrix, it has
a set of real eigenvalues $\lambda _{1}\geq \lambda _{2}\geq ...\geq \lambda
_{n}$. The powers of the adjacency matrix $A_{\mathcal{G}}$ can be related to walks
in $\mathcal{G}$:

\begin{lemma}
\label{lem_Biggs} Let $\mathcal{G}$ be a weighted undirected graph with no
self-loops. The $\left( i,j\right) $-th entry of $A_{\mathcal{G}}^{k}$ can
be written in terms of walks in $\mathcal{G}$ as follows:%
\begin{equation*}
\left[ A_{\mathcal{G}}^{k}\right] _{ij}=\sum_{\mathbf{w}\in W_{i,j}^{\left(
k\right) }}a_{\mathbf{w}},
\end{equation*}%
where $W_{i,j}^{\left( k\right) }$ is the set of all closed walks of length $%
k$ from node $i$ to $j$ in the complete graph $K_{n}\footnote{%
The complete graph $K_{n}$, is the undirected graph with $n$ nodes in which
every pair of distinct vertices is connected by a unique edge.}$.
\end{lemma}

\begin{proof}
The above result comes directly from an algebraic expansion of $\left[ A_{%
\mathcal{G}}^{k}\right] _{v_{0},v_{k}}$. First, notice that%
\begin{equation*}
\left[ A_{\mathcal{G}}^{k}\right] _{v_{0},v_{k}}=%
\sum_{v_{1}=1}^{n}a_{v_{0},v_{1}}\left[ A_{\mathcal{G}}^{k-1}\right]
_{v_{1},v_{k}}.
\end{equation*}%
Using the above rule in a simple recursion, we can expand the entries of the
decreasing powers of $A_{\mathcal{G}}$, to obtain%
\begin{eqnarray*}
\left[ A_{\mathcal{G}}^{k}\right] _{v_{0},v_{k}} &=&\sum_{1\leq
v_{1},...,v_{k-1}\leq n}a_{v_{0},v_{1}}a_{v_{1},v_{2}}...a_{v_{k-1},v_{k}} \\
&=&\sum_{\mathbf{w}\in W_{i,j}^{\left( k\right)
}}\prod_{i=0}^{k-1}a_{v_{i},v_{i+1}},
\end{eqnarray*}%
which is the statement of our lemma.
\end{proof}

The following result will be useful in our derivations:

\begin{lemma}
\label{lem_partial_of_trace} Let $\mathcal{G}$ be a weighted undirected
graph with no self-loops. Then, we have that%
\begin{equation}
\frac{\partial \mathbf{tr}(A_{\mathcal{G}}^{k})}{\partial a_{ij}}=2k\left[
A_{\mathcal{G}}^{k-1}\right] _{ij}.
\end{equation}
\end{lemma}

\begin{proof}
First, notice that for any two matrices $S=\left[ s_{ij}\right] $ and $B=%
\left[ b_{ij}\right] $, we have that $\mathbf{tr}\left( SB\right)
=\sum_{i,j}s_{ij}b_{ji}$. Consider $S$ and $B$ to be symmetric matrices.
Then, we can write $S=U+L$, where $U$ and $L$ are upper and lower triangular
matrices, respectively, with $L=U^{T}$. Let $i<j$ (the same holds for $j<i$%
), hence 
\begin{eqnarray*}
\frac{\partial }{\partial s_{ij}}\mathbf{tr}(SB) &=&\frac{\partial }{%
\partial s_{ij}}\left[ \mathbf{tr}(UB)+\mathbf{tr}(LB)\right] \\
&=&\frac{\partial }{\partial s_{ij}}\left[ \sum_{i<j}u_{ij}b_{ji}+%
\sum_{i>j}l_{ij}b_{ji}\right] \\
&=&\frac{\partial }{\partial u_{ij}}\sum_{i<j}u_{ij}b_{ji}+%
\sum_{j>i}u_{ij}b_{ij} \\
&=&b_{ji}+b_{ij}=2b_{ij}.
\end{eqnarray*}%
Then, we have that

\begin{eqnarray*}
\frac{\partial \mathbf{tr}(A_{\mathcal{G}}^{k})}{\partial a_{ij}} &=&\frac{%
\partial \mathbf{tr}(S_{\mathcal{G}}A_{\mathcal{G}}^{k-1})}{\partial s_{ij}}+%
\frac{\partial \mathbf{tr}(A_{\mathcal{G}}S_{\mathcal{G}}A_{\mathcal{G}%
}^{k-2})}{\partial s_{ij}}+... \\
&&+\frac{\partial \mathbf{tr}(A_{\mathcal{G}}^{k-1}S_{\mathcal{G}})}{%
\partial s_{ij}} \\
&=&k\frac{\partial \mathbf{tr}(S_{\mathcal{G}}A_{\mathcal{G}}^{k-1})}{%
\partial s_{ij}}=2k\left[ A_{\mathcal{G}}^{k-1}\right] _{ij},
\end{eqnarray*}
which completes the proof.
\end{proof}

%%%%%%%%%%%%%%%%%%%%%%%%%%%%%%%%%%%%%%%%%%%%%%%%%%%%%%%%%%%%%%%%%%%%%%%%%

\subsection{Problem Definition}

\label{sec_problem}

Consider a group of $n$ mobile robots and define by $x_{i}(t)\in \mathbb{R}%
^{d}$ the position of robot $i$ at time $t\geq 0$. Let $x=\left[ 
\begin{matrix}
x_{1}^{T} & \dots & x_{n}^{T}%
\end{matrix}%
\right] ^{T}\in \mathbb{R}^{dn}$ denote the stacked column vector of all
robot positions, so that $x_{ir}$ is the $r$-th coordinate of the $i$-th
robot position. We assume that we can control the position of the robots by
the simple kinematic law 
\begin{equation}
\dot{x}=u=\left[ 
\begin{matrix}
u_{1}^{T} & \dots & u_{n}^{T}%
\end{matrix}%
\right] ^{T}\in \mathbb{R}^{dn},  \label{Closed-Loop System}
\end{equation}%
where $u_{i}\left( t\right) \in \mathbb{R}^{d}$ is the control input applied
to robot $i$, and $u_{ir}\left( t\right) $ is the $r$-th coordinate of $%
u_{i}\left( t\right) $.

For a given set of robot positions, $\left\{ x_{i}\right\} _{i=1}^{n}$, we
define the weighted adjacency matrix of the network of robots as $A(x)=\left[
a_{ij}(x)\right] $ with%
\begin{equation}
a_{ij}(x)\triangleq \mathbf{e}^{-c\left\Vert x_{i}-x_{j}\right\Vert _{z}},
\label{eqn_adjacency_entries}
\end{equation}%
where $c>0$ is a constant and $z\in \left\{ 1,2\right\} $ denotes the $\ell
_{1}$ or $\ell _{2}$ norm. Notice that $A(x)$ is a symmetric matrix with
position-dependent real eigenvalues, $\left\{ \lambda _{i}\left( x\right)
\right\} _{i=1}^{n}$. We define, further, the $k$-th spectral moment of $%
A\left( x\right) $ by%
\begin{equation}
m_{k}\left( x\right) \triangleq \frac{1}{n}\sum_{i=1}^{n}\lambda
_{i}^{k}\left( x\right) =\frac{1}{n}\mathbf{tr}\left[ A^{k}\left( x\right) %
\right] ,  \label{eqn_spectral_moments}
\end{equation}%
for $1\leq k\leq n$, where the last equality follows from diagonalizing $%
A^{k}\left( x\right) $. For any finite network with $n$ nodes, its
eigenvalue spectrum is uniquely defined by the sequence of $n$ moments $%
\left( m_{1},m_{2},...,m_{n}\right) $. Therefore, we can simultaneously
control the whole set of eigenvalues of a network by controlling the first $%
n $ spectral moments of $A\left( x\right) $. This gives rise to the
following problem that we aim to address:

\begin{problem}[Control of spectral moments]
\label{problem} Let $\{m_k^\star\}_{k=1}^n$ denote a desired set of spectral
moments. Design control laws $u_i$ for all robots $i=1,\dots,n$ so that the
adjacency matrix $A(x)$ of the position-dependent robot network has spectral
moments that satisfy $m_k(x)\to m_{k}^{\star }$ for all $k=1,\dots,n$.
\end{problem}

Controlling the eigenvalue spectrum of the adjacency matrix of a network is
particularly important, since it is related to the behavior of interesting
network dynamical properties \cite{P08}. In Section~\ref%
{sec_control_convergence} we propose a gradient descent algorithm to address
Problem~\ref{problem} and discuss its convergence properties. In Section~\ref%
{sec_simulations} we illustrate our approach in numerical simulations.

%%%%%%%%%%%%%%%%%%%%%%%%%%%%%%%%%%%%%%%%%%%%%%%%%%%%%%%%%%%%%%%%%%%%%%%%%

\section{Control of Spectral Moments}

\label{sec_control_convergence}

%%%%%%%%%%%%%%%%%%%%%%%%%%%%%%%%%%%%%%%%%%%%%%%%%%%%%%%%%%%%%%%%%%%%%%%%%

\subsection{Controller Design}

\label{sec_control}

Assume a given sequence of desired spectral moments $\{m_{k}^{\star
}\}_{k=1}^{n}$ and define the cost function%
\begin{equation}  \label{eqn_cost_function}
f(x)\triangleq \sum_{k=1}^{n}\frac{1}{4k}\left( m_{k}\left( x\right)
-m_{k}^{\star }\right) ^{2},
\end{equation}
where $m_{k}\left( x\right) $ is defined in (\ref{eqn_spectral_moments}).
Define, further, the gradient descent control law 
\begin{equation}  \label{eqn_gradient_controller}
u\triangleq -\nabla _{x}f(x).
\end{equation}
Then, we can show the following result:

\begin{lemma}
\label{lem_control_law} Let the adjacency matrix $A=A\left( x\right) $ be
defined as in (\ref{eqn_adjacency_entries}) and assume that $z=1$ ($\ell_1$
norm). Then, an explicit expression for the entries of $u$ is given by 
\begin{equation*}
u_{ir}=\frac{c}{n}\sum_{k=1}^{n}\left( \frac{1}{n}\mathbf{tr}%
A^{k}-m_{k}^{\star }\right) \left[ \left( A\circ S_{r}\right) A^{k-1}\right]
_{ii},
\end{equation*}%
where $S_{r}\left( x\right) =\left[ s_{ij}\right] $ with $s_{ij}\triangleq $%
sgn$\left( x_{ir}-x_{jr}\right) $.
\end{lemma}

\begin{proof}
The adjacency matrix of an undirected graph with no self-loops has $\binom{n%
}{2}$ independent entries (for example, the upper triangular entries). Each
one of these entries are a function of the vector of positions $x$. Hence,
applying the chain rule, we have the following expansion for the partial
derivative of the cost function with respect to the entries $x_{ir}$:%
\begin{equation}  \label{eqn_chain_rule}
\frac{\partial f(x)}{\partial x_{ir}}=\sum_{1\leq p<q\leq n}\frac{\partial f%
}{\partial a_{pq}}\frac{\partial a_{pq}}{\partial x_{ir}}.
\end{equation}%
Furthermore, a particular entry $a_{pq}$ depends solely on the position $%
x_{p}$ and $x_{q}$; therefore, we have that $\frac{\partial a_{pq}}{\partial
x_{ir}}=0$ if both $p$ and $q$ are different than $i$. Thus, for a fixed $i$%
, only the following summands in (\ref{eqn_chain_rule}) survive%
\begin{eqnarray}  \label{eqn_surviving_terms}
\frac{\partial f(x)}{\partial x_{ir}} &=&\sum_{p=1}^{i}\frac{\partial f}{%
\partial a_{pi}}\frac{\partial a_{pi}}{\partial x_{ir}}+\sum_{q=i}^{n}\frac{%
\partial f}{\partial a_{iq}}\frac{\partial a_{iq}}{\partial x_{ir}}  \notag
\\
&=&\sum_{j=1}^{n}\frac{\partial f}{\partial a_{ij}}\frac{\partial a_{ij}}{%
\partial x_{ir}},
\end{eqnarray}%
where we have used that $a_{ij}=a_{ji}$ and $a_{ii}=0$.

We now analyze each one of the partial derivatives in (\ref%
{eqn_surviving_terms}). First, from (\ref{eqn_cost_function}) and (\ref%
{eqn_spectral_moments}), we have that%
\begin{eqnarray}
\frac{\partial f}{\partial a_{ij}} &=&\sum_{k=1}^{n}\frac{1}{2k}\left(
m_{k}\left( x\right) -m_{k}^{\star }\right) \frac{\partial }{\partial a_{ij}}%
\left( \frac{1}{n}\mathbf{tr}\left( A^{k}\left( x\right) \right) \right) 
\notag  \label{eqn_partial_f} \\
&=&\frac{1}{n}\sum_{k=1}^{n}\left( m_{k}\left( x\right) -m_{k}^{\star
}\right) \left[ A^{k-1}\left( x\right) \right] _{ij},
\end{eqnarray}%
where we have used\ Lemma \ref{lem_partial_of_trace} in the last equality.
Second, from (\ref{eqn_adjacency_entries}) we have that%
\begin{eqnarray}
\frac{\partial a_{ij}\left( x\right) }{\partial x_{ir}} &=&-c\mathbf{e}%
^{-c\left\Vert x_{i}-x_{j}\right\Vert _{1}}\frac{\partial }{\partial x_{ir}}%
\left\vert x_{ir}-x_{jr}\right\vert  \notag  \label{eqn_partial_a} \\
&=&-ca_{ij}\text{sgn}\left( x_{ir}-x_{jr}\right) ,
\end{eqnarray}%
for $x_{ir}\neq x_{jr}$ (the fact that $\left\vert x_{ir}-x_{jr}\right\vert $
is not differentiable at $x_{ir}=x_{jr}$ does not affect our analysis). Let
us define the antisymmetric matrix $S_{r}\left( x\right) =\left[
s_{ij}^{\left( r\right) }\left( x\right) \right] $ with $s_{ij}^{\left(
r\right) }\left( x\right) \triangleq $sgn$\left( x_{ir}-x_{jr}\right) $.
Hence, we can write (\ref{eqn_partial_a}) using a Hadamard product as follows%
\begin{equation}
\frac{\partial a_{ij}\left( x\right) }{\partial x_{ir}}=-c\left[ A\left(
x\right) \circ S_{r}\left( x\right) \right] _{ij}.
\label{eqn_partial_a_nice}
\end{equation}%
Substituting (\ref{eqn_partial_f}) and (\ref{eqn_partial_a_nice}) in (\ref%
{eqn_surviving_terms}), we have%
\begin{eqnarray*}
\frac{\partial f(x)}{\partial x_{ir}} &=&-\frac{c}{n}\sum_{k=1}^{n}\left(
m_{k}\left( x\right) -m_{k}^{\star }\right) \sum_{j=1}^{n}\left[ A^{k-1}%
\right] _{ji}\left[ A\circ S_{r}\right] _{ij} \\
&=&-\frac{c}{n}\sum_{k=1}^{n}\left( m_{k}\left( x\right) -m_{k}^{\star
}\right) \left[ \left( A\circ S_{r}\right) A^{k-1}\right] _{ii},
\end{eqnarray*}%
where both matrices $A$ and $S_{r}$ depend on $x$. Then, from (\ref%
{eqn_spectral_moments}), we obtain the statement of our lemma.
\end{proof}

An efficient relaxation of the spectral control Problem~\ref{problem}
results from controlling a truncated sequence of spectral moments $\left(
m_{1},...,m_{s}\right) $, for $s<n$. In this case, we can define a cost
function%
\begin{equation}  \label{eqn_alternative_cost_function}
f_{s}(x)=\sum_{k=1}^{s}\frac{1}{4k}\left( m_{k}\left( x\right) -m_{k}^{\star
}\right) ^{2},
\end{equation}%
and an associated control law $u=-\nabla _{x}f_{s}(x)$. An explicit
expression for $u$ can be obtained by following the steps in the proof of
Lemma \ref{lem_control_law}, which result in%
\begin{equation}  \label{eqn_alternative_control_law}
u_{ir}^{\left( s\right) }=\frac{c}{n}\sum_{k=1}^{s}\left( \frac{1}{n}\mathbf{%
tr}A^{k}-m_{k}^{\star }\right) \left[ \left( A\circ S_{r}\right) A^{k-1}%
\right] _{ii}.
\end{equation}

Although controlling a truncated sequence of moments is not mathematically
equivalent to controlling the whole eigenvalue spectrum of the adjacency, we
observe a very good overall matching between the eigenvalues obtained from
the relaxed problem and the desired eigenvalue spectrum, especially for the
eigenvalues of largest magnitude, which are usually the most relevant in
dynamical problems (Section~\ref{sec_simulations}).

%%%%%%%%%%%%%%%%%%%%%%%%%%%%%%%%%%%%%%%%%%%%%%%%%%%%%%%%%%%%%%%%%%%%%%%%%

\subsection{Convergence Analysis}

\label{sec_convergence}

To simplify convergence analysis of the closed-loop system (\ref{Closed-Loop System}%
), we restrict its dynamics  to the open set $%
\mathcal{F}=\{x\;|\;m_{k}\left( x\right) >m_{k}^{\star },\;\;\forall \;1\leq
k\leq s\}$.\footnote{It is shown in Theorem~\ref{thm_convergence} that this construction does not restrict convergence of the state variables to the desired equilibria.} We can ensure that $x\in\mathcal{F}$ for all time by adding the barrier
potential%
\begin{equation}
b_{s}(x)\triangleq \sum_{k=1}^{s}\frac{\epsilon _{k}}{4k}\frac{1}{\left(
m_{k}\left( x\right) -m_{k}^{\star }\right) ^{2}},
\label{eqn_barrier_function}
\end{equation}%
to (\ref{eqn_cost_function}), for sufficiently small constants $\epsilon _{1},\dots ,\epsilon _{s}>0$
(assuming that the initial state of the system is already in $\mathcal{F}$%
). The convergence properties of the resulting closed loop system $\dot{x}%
=-\nabla _{x}(f_{s}(x)+b_{s}(x))$ are discussed in the following result.

\begin{theorem}
\label{thm_convergence} Let $\{m_{k}^\star\}_{k=1}^s$ denote a desired set
of adjacency spectral moments and assume that $x(0)\in\mathcal{F}$. Then,
for sufficiently small $\epsilon_1\dots,\epsilon_s>0$, the closed loop
system 
\begin{equation}  \label{eqn_closed_loop}
\dot{x}=-\nabla _{x}\left(f_s(x)+b_s(x)\right)
\end{equation}%
ensures that the moments $\{m_k(\tilde{x})\}_{k=1}^s$, where $\tilde{x}%
=\lim_{t\rightarrow \infty }x(t)$, approximate arbitrarily well the desired
set $\{m_k^\star\}_{k=1}^s$.
\end{theorem}

\begin{proof}
The time derivative of $f_{s}(x)+b_{s}(x)$ is given by 
\begin{eqnarray}
\frac{d}{dt}(f_{s}(x)+b_{s}(x)) &=&\nabla _{x}(f_{s}(x)+b_{s}(x))\dot{x} 
\notag  \label{eqn_time_derivative} \\
&=&-\Vert \nabla _{x}(f_{s}(x)+b_{s}(x))\Vert _{2}^{2}\leq 0
\end{eqnarray}%
and so the closed loop system is stable and $x$ will converge to a minimum
of $f_{s}(x)+b_{s}(x)$. Since $b_{s}(x)\rightarrow \infty $ whenever $%
x\rightarrow \partial \mathcal{F}$, where $\partial \mathcal{F}%
=\{x\;|\;m_{k}\left( x\right) =m_{k}^{\star },\;\;\forall \;1\leq k\leq
s\}=\{x\;|\;f_{s}(x)=0\}$ denotes the boundary of the set $\mathcal{F}$,
equation \eqref{eqn_time_derivative} also implies that the set $\mathcal{F}$
is an invariant of motion for the system under consideration. Let 
\begin{equation*}
\mathcal{P}_{0}=\{x\;|\;x_{ir}(0)\leq x_{jr}(0)\;\Rightarrow \;x_{ir}\leq
x_{jr},\;\forall \;i,j,r\}
\end{equation*}%
denote the polytope defined by the relative positions of the robots with
respect to their initial configuration $x(0)$, so that $x(0)\in \mathcal{P}%
_{0}\cap \mathcal{F}$ (see Fig.~\ref{fig_polytope}\footnote{%
The polytope $\mathcal{P}_{0}$ essentially defines an ordering of the state
variables.}). In what follows we show that as $\epsilon _{1},\dots ,\epsilon
_{s}\rightarrow 0$, the state variable $x$ asymptotically reaches a value in
the set $\mathcal{P}_{0}\cap \partial \mathcal{F}$, where $f_{s}(x)=0$.

To see this, observe first that 
\begin{equation*}
\lim_{\epsilon_1,\dots,\epsilon_s\to 0}(f_s(x)+b_s(x))=f_s(x)
\end{equation*}
for all $x\in\mathcal{P}_0\cap\mathcal{F}$, i.e., for small enough $%
\epsilon_1,\dots,\epsilon_s>0$, the potential $f_s(x)+b_s(x)$ can be
approximated by $f_s(x)$ in $\mathcal{P}_0\cap\mathcal{F}$. Moreover, note
that $f_s(x)$ is convex in the set $\mathcal{P}_0\cap\mathcal{F}$. Convexity
of $f_s(x)$ follows from the fact that $-c|x_{ir}-x_{jr}|$ is affine in $%
\mathcal{P}_0$ for any $c>0$ and, therefore, $-\sum c|x_{ir}-x_{jr}|$ is
also affine for any number of terms in the summation. This implies that $%
e^{-\sum c|x_{ir}-x_{jr}|}$ is convex in $\mathcal{P}_0$ as a composition of
a convex and an affine function and, therefore, 
\begin{eqnarray*}
\frac{1}{n}\mathbf{tr}A^k(x) -m_k^\star &=& \frac{1}{n}\sum \prod
e^{-c|x_{ir}-x_{jr}|} -m_k^\star \\
&=& \frac{1}{n}\sum e^{-\sum c|x_{ir}-x_{jr}|} -m_k^\star
\end{eqnarray*}
is also convex as a sum of convex functions. Clearly, $\frac{1}{n}\mathbf{tr}%
A^k(x)-m_k^\star$ is nonnegative for any $x\in\mathcal{P}_0\cap\mathcal{F}$.
Since any power greater than one of a nonnegative and convex function is
also convex, every one of the terms $\left(\frac{1}{n}\mathbf{tr}%
A^k(x)-m_k^\star\right)^2$, for $k=1,\dots,s$, is convex, which implies that 
$f_s(x)$ is also convex in $\mathcal{P}_0\cap\mathcal{F}$. Taking the limit
of \eqref{eqn_time_derivative} we have that 
\begin{equation*}
\lim_{\epsilon_1,\dots,\epsilon_s}\frac{d}{dt}(f_s(x)+b_s(x)) = \dot{f}_s(x)
\leq 0
\end{equation*}
in $\mathcal{P}_0\cap\mathcal{F}$. Therefore, convexity of $f_s(x)$ along
with the condition $\dot{f}_s(x)\leq 0$ implies that $x$ will converge to a
global minimum of $f_s(x)$ in $\mathcal{P}_0\cap\mathcal{F}$ (recall that $%
\mathcal{F}$ is an invariant of motion for the system under consideration).
Since 
\begin{equation*}
\inf_{x\in\mathcal{P}_0\cap(\mathcal{F}\cup\partial\mathcal{F}%
)}\{f_s(x)\}=\left.f_s(x)\right|_{x\in\mathcal{P}_0\cap\partial\mathcal{F}%
}=0,
\end{equation*}
we conclude that the system will converge to a network with spectral moments 
$\{m_k(x)\}_{k=1}^s$ that are almost equal to the desired $%
\{m_k^\star\}_{k=1}^s$. The quality of the approximation depends on how
small the constants $\epsilon_1,\dots,\epsilon_s>0$ are.

\begin{figure}[t]
\centering
\includegraphics[width=.9\linewidth]{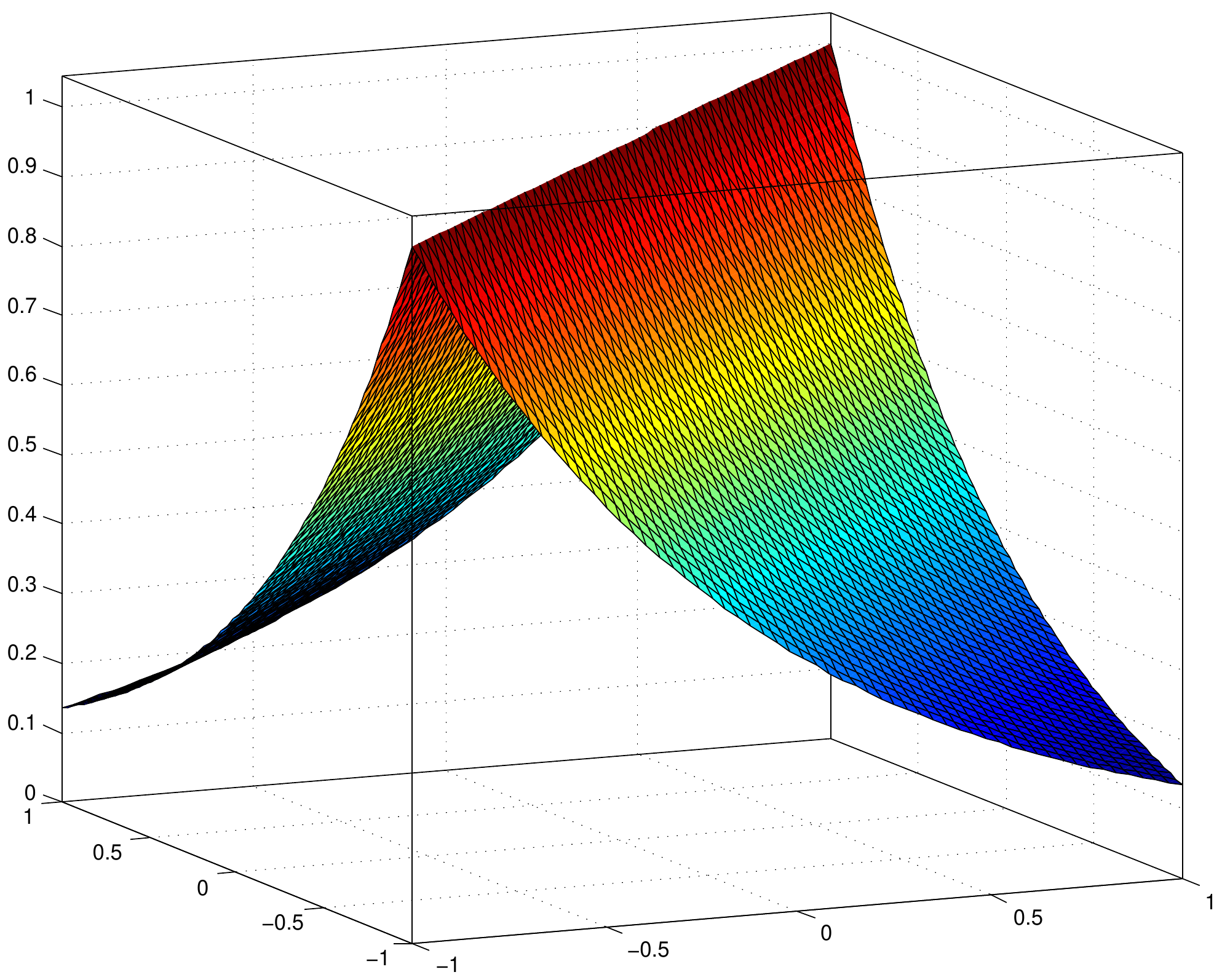}\newline
\caption{Plot of the function $e^{-|x_{ir}-x_{jr}|}$. Note the boundary $%
x_{ir}=x_{jr}$ of two adjacent convex polytopes $\mathcal{P}$.}
\label{fig_polytope}
\end{figure}

What remains is to show that $\partial\mathcal{F}\subset\mathcal{P}_0$. For
this, assume that $\partial\mathcal{F}\not \subset\mathcal{P}_0$. Then, if
the desired set of moments $\{m_k^\star\}_{k=1}^s$ is realizable, there
exists another polytope $\mathcal{P}^\star\neq\mathcal{P}_0$ such that $%
\partial\mathcal{F}\subset\mathcal{P}^\star$. Equivalently, there exists a
configuration $x^\star\in\mathcal{P}^\star$ such that $A(x^\star)$ has
eigenvalues $\{\lambda_i(x^\star)\}_{i=1}^n$ with $\sum_{i=1}^n
\lambda_i^k(x^\star)=m_k^\star$ for all $k=1,\dots,s$. The result follows
from the observation that $\mathcal{P}_0$ can be obtained from $\mathcal{P}%
^\star$ by changing the relative positions of the robots in $\mathbb{R}^d$.
Mathematically, this means that there exists a permutation matrix $\Pi\in%
\mathbb{R}^{nd\times nd}$, i.e., an orthogonal matrix with 0 or 1 entries,
such that that if $\mathcal{P}^\star=\{x \; | \; Ex\geq 0\}$ with $E\in%
\mathbb{R}^{n^2d\times nd}$, then $\mathcal{P}_0=\{x \; | \; E(\Pi x) \geq
0\}$. Therefore, if $x^\star\in\mathcal{P}^\star$, then $\Pi^Tx^\star\in%
\mathcal{P}_0$. Moreover, 
\begin{equation*}
A(\Pi^Tx^\star)\otimes I_d = \Pi (A(x^\star)\otimes I_d)\Pi^T,
\end{equation*}
i.e., a permutation of the robots' coordinates results in a permutation of
the entries of $(A(x)\otimes I_d)$, where $\otimes$ indicates the Kronecker
product between matrices. Since $\Pi$ is orthogonal, $A(\Pi^Tx^\star)\otimes
I_d$ and $A(x^\star)\otimes I_d$ have the same eigenvalues. Therefore, there
exists a configuration $\tilde{x}=\Pi^Tx^\star\in\mathcal{P}_0$ such that $A(%
\tilde{x})$ has eigenvalues $\{\lambda_i(\tilde{x})\}_{i=1}^n$ with $%
\sum_{i=1}^n \lambda_i^k(\tilde{x})=m_k^\star$ for all $k=1,\dots,s$.%
\footnote{%
Equivalently, this means that the two graphs are isomorphic.} This implies
that $\partial\mathcal{F}\subset\mathcal{P}_0$.

In the above discussion, we have used the fact that if $\{y_{i}\}_{i=1}^{n}$
are the eigenvalues of $Y\in \mathbb{R}^{n\times n}$ and $%
\{z_{i}\}_{i=1}^{m} $ are the eigenvalues of $Z\in \mathbb{R}^{m\times m}$,
then the eigenvalues of $Y\otimes Z$ are $\{y_{1}z_{1},\dots
,y_{1}z_{m},\dots ,y_{n}z_{1},\dots ,y_{n}z_{m}\}$. This implies that the
eigenvalues of $A(x)\otimes I_{d}$ are essentially the eigenvalues of $A(x)$%
, each one with multiplicity $d$.
\end{proof}

\begin{remark}[Barrier functions]
\label{rem_barrier_functions} The barrier functions $b_s(x)$ are necessary
in the proof of Theorem~\ref{thm_convergence}. If not there, the quantities $%
m_k(x) -m_k^\star$ are not guaranteed to be positive and, therefore, the
potential $f_s(x)$ is not necessarily convex in $\mathcal{P}_0$. This causes
technical difficulties in ensuring a global minimum of $f_s(x)$. Also, the
analysis in Theorem~\ref{thm_convergence} assumes that $\epsilon_1,\dots,%
\epsilon_s\to 0$. This essentially restricts the influence of the barrier
potential $b_s(x)$ to a small neighborhood of the set $\partial\mathcal{F}$,
so that the potential $f_s(x)$ remains unaffected outside this neighborhood.
In practice, the closer the constants $\epsilon_1,\dots,\epsilon_s$ are to
zero, the better the approximation of the desired spectral moments. This
approximation can be made arbitrarily good.
\end{remark}

\begin{remark}[Distance metric]
\label{rem_distance_metric} In the preceding analysis, we have employed the $%
\ell_1$ norm ($z=1$) as a distance metric to define the entries of the
adjacency matrix (see (\ref{eqn_adjacency_entries})). This choice is mainly
due to technical reasons, since it ensures that $\mathbf{tr}A^k(x)$ and,
therefore, $f_s(x)$, are convex functions. From a practical point of view,
the $\ell_1$ metric has potential applications in indoor navigation where
the presence of obstacles forces the signals to propagate in grid-like
environments. Nevertheless, our numerical simulations indicate that the
control law herein proposed also performs well for $z=2$ ($\ell_2$ norm). We
leave a rigorous proof of this case for future work.
\end{remark}

%%%%%%%%%%%%%%%%%%%%%%%%%%%%%%%%%%%%%%%%%%%%%%%%%%%%%%%%%%%%%%%%%%%%%%%%%

\section{Numerical Simulations}

\label{sec_simulations}

\begin{figure}[t]
\centering
\includegraphics[width=.80\linewidth]{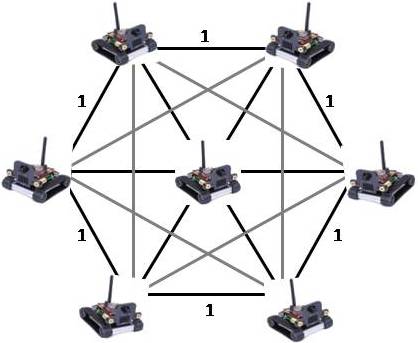}
\caption{Hexagonal formation of mobile robots.}
\label{fig_hex_net}
\end{figure}

In this section we present examples of spectral design of mobile robot
networks and discuss performance of our proposed approach.

\begin{example}[Hexagonal Formation]
\label{example_hexagonal} Consider the problem of controlling the structure
of a network of mobile robots to match the eigenvalue spectrum of the
hexagonal network on 7 nodes shown in Fig. \ref{fig_hex_net}. The target
eigenvalues of the weighted adjacency matrix of this formation are 
\begin{equation*}
\left\{ \lambda _{i}^{\star }\right\} _{i=1}^{7}=\left\{
-0.51, \;-0.47, \;-0.40, \;-0.40, \;0.05, \;0.05, \; 1.70\right\} ,
\end{equation*}%
and the corresponding spectral moments are%
\begin{equation*}
\left\{ m_{k}^{\star }\right\} _{k=1}^{7}=\left\{
0, \;0.53, \;0.64, \;1.22, \;2.02, \;3.47, \;5.90\right\} .
\end{equation*}%
The initial configuration of the mobile robot network is that of a random
geometric graph on $7$ nodes uniformly distributed in the square $\left[ 0,1%
\right] \times \left[ 0,1\right] $. We applied the proposed control law %
\eqref{eqn_closed_loop} and studied the evolution of the spectral moments of
the network's adjacency matrix. Fig.~\ref{fig_hex_conv} shows the evolution
of the second, third and fourth moment.\footnote{%
A similar behavior is observed for the higher order moments. Note that the
first spectral moment is always equal to zero for simple graphs.} As
expected, they all converge to the desired values. The asymptotic values of
all spectral moments of the network are%
\begin{equation*}
\left\{ m_{k}\right\} _{k=1}^{7}=\left\{
0, \;0.53, \;0.65, \;1.23, \;2.04, \;3.50, \;5.95\right\} ,
\end{equation*}%
which are very close to the desired sequence of moments $\left\{
m_{k}^{\star }\right\} _{k=1}^{7}$, and slightly on the larger side, as predicted by
Theorem \ref{thm_convergence}. The eigenvalues of the weighted adjacency
matrix of the final configuration are 
\begin{equation*}
\left\{ \lambda _{i}\right\} _{i=1}^{7}=\left\{
-0.52, \;-0.48, \;-0.42, \;-0.40, \;0.02, \;0.10, \;1.70\right\} ,
\end{equation*}%
which are also very close to the desired values. Notice that our approach
fits better those eigenvalues further away from the origin, since they are
more heavily weighted in the expression of spectral moments, than those
close to zero. An alternative approach to overcome this limitation would be
to modify our cost function (\ref{eqn_cost_function}) to assign more weight
to eigenvalues of small magnitude.

As discussed in Section~\ref{sec_control}, we also consider a relaxation of
Problem~\ref{problem}, that involves a truncated sequence of the first four
moments of the network. In this case, the cost function is given by (\ref%
{eqn_alternative_cost_function}) and the closed loop system is defined in %
\eqref{eqn_closed_loop}, for $s=4$. The asymptotic values of the first four
spectral moments of the network are 
\begin{equation*}
\left\{ m_{k}\right\} _{k=1}^{4}=\left\{ 0, \;0.55, \;0.65, \;1.26\right\}
\end{equation*}%
and eigenvalues of the final configuration are 
\begin{equation*}
\left\{ \lambda _{i}\right\} _{i=1}^{7}=\left\{
-0.55, \;-0.50, \;-0.42, \;-0.31, \;-0.09, \;0.18, \;1.71\right\} ,
\end{equation*}%
which are remarkably close to the set of desired eigenvalues, especially
those far from zero. The main advantages of using a relaxation of Problem~%
\ref{problem} are (\emph{i}) that it reduces the computational cost of the
controller, since it only requires the terms $\mathbf{tr}A^{k}$ and $A^{k-1}$%
, and (\emph{ii}) that it improves the stability of the numerical behavior
of the gradient descent algorithm.
\end{example}

\begin{figure}[tbp]
\centering
\includegraphics[width=.95\linewidth]{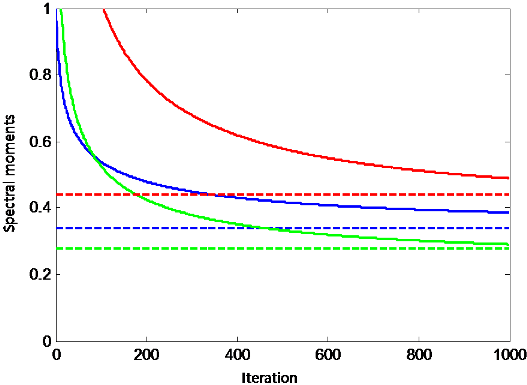}
\caption{Evolution of second (blue), third (green) and fourth (red) spectral
moments towards the moments of the hexagonal formation.}
\label{fig_hex_conv}
\end{figure}

\begin{example}[Random Geometric Network]
\label{example_delta} In this example we consider the problem of matching
the eigenvalue spectrum of a particular realization of a random geometric
graph on $10$ nodes that are uniformly distributed in $\left[ 0,1\right]
\times \left[ 0,1\right] $. This target realization is illustrated in Fig. %
\ref{fig_RGG}, where the thickness of each edge is proportional to its
weight. The eigenvalues of the weighted adjacency matrix of this network are 
\begin{eqnarray*}
\left\{ \lambda _{i}^{\star }\right\} _{i=1}^{10} &=&\left\{
-0.89, \;-0.85, \;-0.84, \;-0.79, \;-0.77, \;-0.68,\right.  \\
&&\left. -0.61, \;0.02, \;0.27, \;5.16\right\} ,
\end{eqnarray*}%
and the corresponding sequence of spectral moments is 
\begin{equation*}
\left\{ m_{k}^{\star }\right\} _{k>0}=\left\{
0, \;3.11, \;13.45, \;71.60, \;368.36, \;1905,\dots \right\} .
\end{equation*}%
We applied the control law \eqref{eqn_closed_loop} using the cost function %
\eqref{eqn_cost_function} and starting with a different realization of the
random geometric graph on $10$ nodes in the square $\left[ 0,1\right] \times %
\left[ 0,1\right] $. In this case, matching the whole set of spectral moments
directly provides us with a very good approximation of the eigenvalues spectrum. As before, we also
considered the relaxation of Problem~\ref{problem} involving only a
truncated sequence consisting of the first 4 moments. The evolution of the
second, third and fourth moments of the network is shown in Fig. \ref%
{fig_RGG_conv}. The asymptotic values of the spectral moments of the network
are 
\begin{equation*}
\left\{ m_{k}\right\} _{k>0}=\left\{ 0, \;3.13, \;13.46, \;71.63,\dots \right\} ,
\end{equation*}%
which are remarkably close to the desired ones and, as in the previous example, slightly on the larger side. Similarly, the final set of
eigenvalues is%
\begin{eqnarray*}
\left\{ \lambda _{i}^{\star }\right\} _{i=1}^{10} &=&\left\{
-0.94, \;-0.87, \;-0.85, \;-0.83, \;-0.74, \;-0.59,\right.  \\
&&\left. -0.55, \;-0.27, \;0.50, \;5.16\right\} ,
\end{eqnarray*}%
which is a very good fit to the given eigenvalues, especially for those far
from zero.
\end{example}

\begin{figure}[tbp]
\centering
\includegraphics[width=0.9\linewidth]{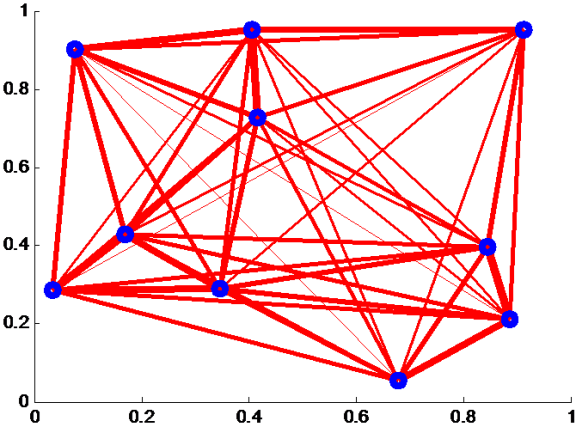}
\caption{Random geometric graph with 10 nodes. The width of each edge is
proportional to its weight.}
\label{fig_RGG}
\end{figure}

%%%%%%%%%%%%%%%%%%%%%%%%%%%%%%%%%%%%%%%%%%%%%%%%%%%%%%%%%%%%%%%%%%%%%%%%%

\section{Conclusions and Future Research}

\label{sec_conclusions_future_research}

In this paper, we proposed a novel control framework to modify the structure
of a mobile robot network in order to control the eigenvalue spectrum of its
adjacency matrix. We introduced a novel abstraction of the eigenvalue
spectrum by means of the spectral moments and
derived explicit gradient descent motion controllers for the robots to obtain a network with the desired set of moments. Since the eigenvalue spectrum is uniquely determined by the associated spectral moments, our approach provides a way of controlling the eigenvalues of mobile networks. Convergence to the desired moments was always guaranteed due to convexity of the proposed cost functions. Efficiency of our approach was illustrated in nontrivial computer simulations. The adjacency
matrix eigenvalue spectrum is relevant to the performance of many
distributed coordination algorithms run over a network. Therefore, our
approach is particularly useful in providing network structures that are
optimal with respect to networked coordination objectives.

Future work involves theoretical guarantees for the Euclidean distance
metric ($z=2$ in (\ref{eqn_adjacency_entries})), as well as extension of our
results to the spectrum of Laplacian matrix of the network. Moreover, the
relaxation of Problem~\ref{problem} to a truncated sequence of moments does
not guarantee (mathematically) a good fit of the complete distribution of
eigenvalues. Therefore, a natural question is to characterize the set of
graphs most of whose spectral information is contained in a relatively small
set of low-order moments.

\begin{figure}[tbp]
\centering
\includegraphics[width=0.9\linewidth]{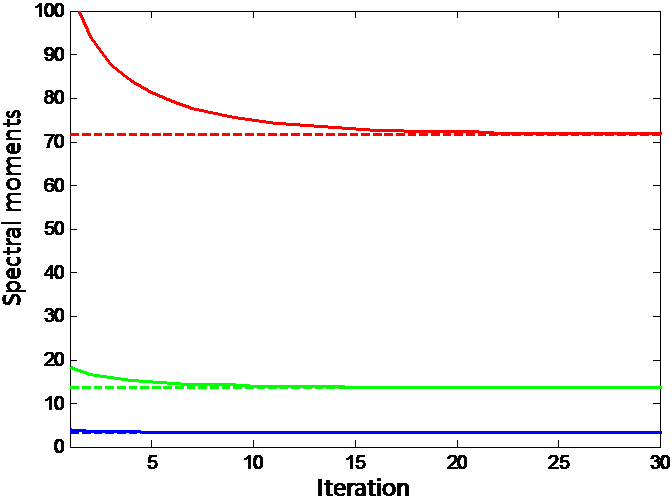}
\caption{Evolution of second (blue), third (green) and fourth (red) spectral
moments towards the moments of the random geometric graph.}
\label{fig_RGG_conv}
\end{figure}

\addtolength{\textheight}{-12cm}

%%%%%%%%%%%%%%%%%%%%%%%%%%%%%%%%%%%%%%%%%%%%%%%%%%%%%%%%%%%%%%%%%%%%%%%%%

\end{document}